%% file: main.tex
\pdfoutput=1 
\documentclass[]{birkjour}
\include{definitions.tex}

\usepackage{hyperref}
\usepackage{cleveref}
\usepackage[font=small,labelfont=bf]{caption}
\usepackage{orcidlink}

\newtheorem{theorem}{Theorem}

\theoremstyle{definition}
\newtheorem{definition}{Definition}[section]

\theoremstyle{remark}

\title{From Invariant Decomposition to Spinors
}

\author[M. Roelfs]{Martin Roelfs \orcidlink{0000-0002-8646-7693}}
\address{%
    {Department of Mathematics} \\
    {University of Antwerp} \\
    {Antwerpen} \\
    {Belgium}
}
\email{martin.roelfs@uantwerpen.be}
\author[D. Eelbode]{David Eelbode}
\address{%
    {Department of Mathematics} \\
    {University of Antwerp} \\
    {Antwerpen} \\
    {Belgium}
}
\email{david.eelbode@uantwerpen.be}
\author[S. De Keninck]{Steven De Keninck \orcidlink{0000-0002-6870-1714}}
\address{%
  Informatics Institute \\
  University of Amsterdam \\
  Amsterdam \\
  The Netherlands
}
\email{steven@enki.ws}

\date{}

\begin{document}

\begin{abstract}
    Plane-based Geometric Algebra (PGA) has revealed points in a $d$-dimensional pseudo-Euclidean space \GR{p,q,1} to be represented by $d$-blades rather than vectors. 
    This discovery allows points to be factored into $d$ orthogonal hyperplanes, establishing points as pseudoscalars of a local geometric algebra \GR{pq}. 
    Astonishingly, the non-uniqueness of this factorization reveals the existence of a local $\Spin{p,q}$ geometric gauge group at each point. 
    Moreover, a point can alternatively be factored into a product of the elements of the Cartan subalgebra of $\spin{p,q}$, which are traditionally used to label spinor representations.
    Therefore, points reveal previously hidden geometric foundations for some of quantum field theory's mysteries.
    This work outlines the impact of PGA on the study of spinor representations in any number of dimensions, and is the first in a research programme exploring the consequences of this insight.
\end{abstract}

\keywords{Spinors, Invariant decomposition, Geometric gauge, Lie groups}

\maketitle

\section{Introduction}
While some of the mystery surrounding spinors is undoubtedly true quantum weirdness, there are also some misunderstood properties of spinors that are utterly classical.
Chief amongst these is the infamous 720 degrees of rotation property, and even some local gauge groups might be purely geometrical in origin.
In order to illustrate that these are classical properties, a quick introduction 
into the reflection mindset of Plane-based Geometric Algebra (PGA) is in order.

PGA embraces the fact that reflections are the atoms of geometry, an insight immortalized in the Cartan-Dieudonn\'{e} theorem: in $d$ dimensions, any isometry can be described as the composition of at most $d$ reflections in hyperplanes.
Combined with the insight that hyperplanes are defined by linear equations in any number of dimensions, one directly identifies the vectors of a Clifford algebra with \emph{hyperplanes}, and not with arrows pointing at points. 
The resulting algebra naturally identifies all elements of geometry as invariants of the isometries they cause when acting as transformations \cite{GSG}.

As an example, \cref{fig:elements_of_geometry} shows all elements of 3DPGA. 
Firstly, a reflection in a plane leaves a single plane invariant, and thus planes are represented by vectors (1-blades).
Secondly, reflecting in two orthogonal planes is identical to a single line-reflection, and thus lines are represented by 2-blades.
Finally, reflecting in three orthogonal planes is identical to a single point-reflection, and thus points are represented by 3-blades.
Moreover, reflections in planes can more generally be combined into translations or rotations (bireflections), transflections or rotoreflections (trireflections), and screws (quadreflections). More generally, such elements are referred to as $k$-reflections.
This mindset raises an interesting question: if points, lines, and planes are all \emph{physical}, in the sense that points, lines, and planes are meaningful objects in the phy\-sical world, then are arbitrary bireflections, trireflections, etc. also physical?
It is our hypothesis that they are indeed naturally occurring, and these are the mysterious objects known as spinors.

This conference proceeding will be the first in a series of manuscripts in which we aim to lay out the case for this hypothesis.
The current manuscript will aim to give an intuitive justification for this claim, while being necessarily brief about some of the details, which will be the subject of future work.

This manuscript is organized as follows: first, we give an overview of several important geometrical principles for understanding spinors in \cref{sec:geometry101}: the first postulate, geometric gauges, double covers, and invariant decomposition. Second, \cref{sec:spinors} starts by giving an introduction into algebraic spinors, and then shows how a very similar story can be told using $k$-reflections. \Cref{sec:spinors} concludes with a brief comparison with Hestenes spinors in \cref{sec:hestenes_spinors}.

\begin{figure}
    \centering
    \includegraphics[width=\textwidth]{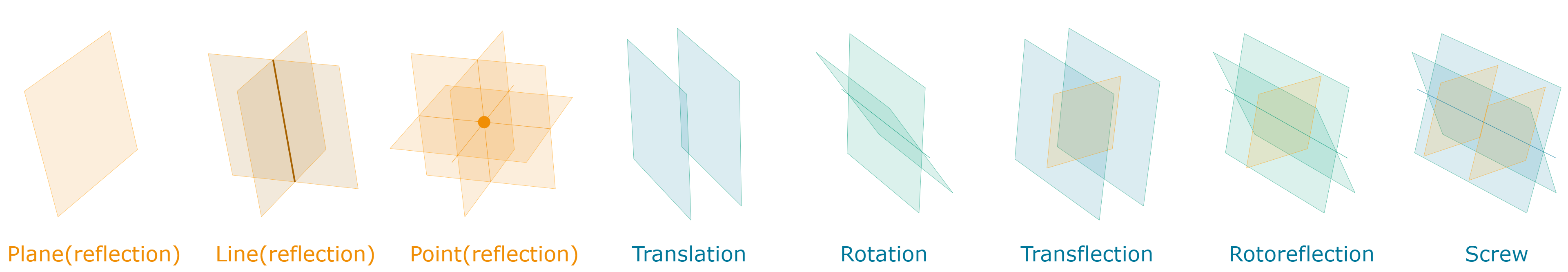}
    \caption{All products of reflections in 3DPGA (\GR{3,0,1}). From left to right: plane(reflection), represented by a vector $u \in \GR[(1)]{3,0,1}$, line(reflection), represented by two orthogonal plane-reflections $u \wedge v \in \GR[(2)]{3,0,1}$, point(reflection), represented by three orthogonal plane-reflections $u \wedge v \wedge w \in \GR[(3)]{3,0,1}$. More generally, plane-reflections are the atoms of all the orthogonal transformations in a space.}
    \label{fig:elements_of_geometry}
\end{figure}

\section{Geometry 101}\label{sec:geometry101}

\subsection{First Postulate Violation and the Principle of Covariance}
\begin{figure}[b]
    \centering
    \includegraphics[width=\textwidth]{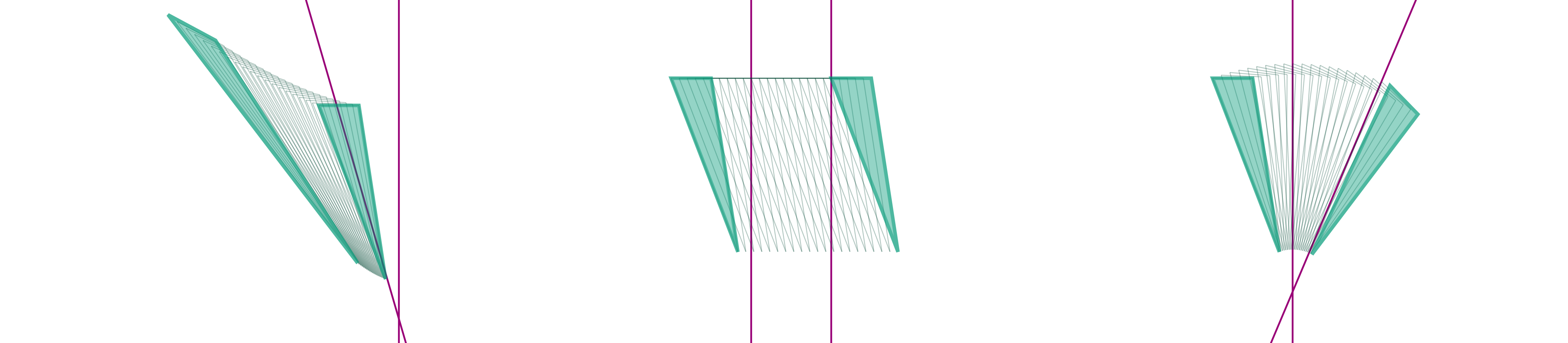}
    \caption{The three fundamental types of bireflection. Left) two reflection in line-time make a boost. Middle) two reflections in parallel lines make a translation. Right) Two reflections in intersecting lines in the plane make a rotation.}
    \label{fig:bireflections}
\end{figure}
Euclid's first postulate essentially states that ``two points define a single line''. However, this statement should have read 
\begin{displayquote}
    ``Two points define a single line, and nothing [...] else!'' \cite{GAME23Steven}
\end{displayquote}
The latter clause implies that any binary operation between two points, when it produces a point, should always produce a point on that line. 
For example, the sum of two points of equal weight should be their midpoint, 
and two points can be joined into a unique line.
In the ``vectors are arrows'' mindset, \emph{first postulate compliance} is not at all guaranteed.
This is easily illustrated by the addition of two arrows pointing at points, which in general does not produce an arrow pointing at the midpoint.
Because of this, the scientific literature is instead full of \emph{first postulate violations}.
With the plane-based mindset however, first postulate compliance is built into the framework.
And once done correctly at the level of points, it automatically extends to lines, planes, etc., and the interactions between these elements.

An important consequence of the first postulate is that it is the elements which determine the dimension of a problem, not the algebra. 
For example, two reflections in intersecting planes define a rotation around the intersection line. 
But there is a vantage point in which this intersection line looks like a point, and the planes look like lines.
Any bireflection can therefore be made to look two-dimensional (\cref{fig:bireflections}). Simply put,
    \begin{displayquote}
        ``Everything is as simple as you can make it look.''
    \end{displayquote}
So we should take a dimension-agnostic approach to geometry
and moreover, we should treat even and odd elements on equal footing, since e.g. points in even/odd numbers of dimensions are represented using even/odd elements respectively.
So true adherence to the ``Principle of Covariance'' implies a formulation of physical laws using those physical quantities on which observers with frames of reference \emph{of possibly different numbers of dimensions} could unambiguously correlate. 
The surprising thing is that plane-based geometric algebra, with its in-built first postulate compliance, allows us to do exactly that.

The first postulate has direct consequences for the study of spinors, 
because one typically starts by choosing the dimension of the space, and then describes the spinor representations in that space.
But when the chosen dimension is odd there is only a single irreducible Dirac representation of spinors, whereas if the chosen dimension is even the now reducible Dirac representation can be decomposed into two irreducible Weyl representations.
But how can this arbitrary cutoff dimension have such severe consequences for the properties of spinors?
Surely the first postulate demands that the Weyl spinor components of 2D spinors should still appear like Weyl spinors when viewed from 17 dimensions?
This is a severe first postulate violation, one which we will aim to avoid.

\subsection{Geometric Gauges}\label{sec:gauge}

Before we dive into the algebraic details of the invariant decomposition, let us spend a paragraph to examine its geometric significance and origin. Every element of a Pin group can be written as a product of 
reflections.
The composition of one, two and three reflections in the plane is depicted in \cref{fig:reflections123}, where the arrows indicate the effect of the transformation on point probes in the space. In orange we show the main invariant. 
On the left we see the effect of a single reflection $u$.
All points in space are reflected to the other side of the mirror, except those that lie on the mirror, which remain in place.
Hence the orange and green lines coincide.
The relationship between the group element and its invariant is clear, and this remains true when we consider the center image, depicting the composition of two reflections $u$ and $v$. 
Here the resulting group element is a rotation, 
and all points in space rotate around the intersection point $u \wedge v$ of the two mirrors.
The only invariant point is the intersection point itself.
There is a second hidden/obvious invariant: the entire plane itself.
When three reflections are composed, as shown on the right, not a single point on the plane is left invariant. 
Nonetheless, there is again an invariant line, shown in orange.
However, there is no longer a clear relation between the orange invariant line and the reflections that were used to construct the transformation, so where did this line come from?
The answer is \emph{geometric gauges}.

\begin{figure}[h]
    \centering
    \includegraphics[width=\textwidth]{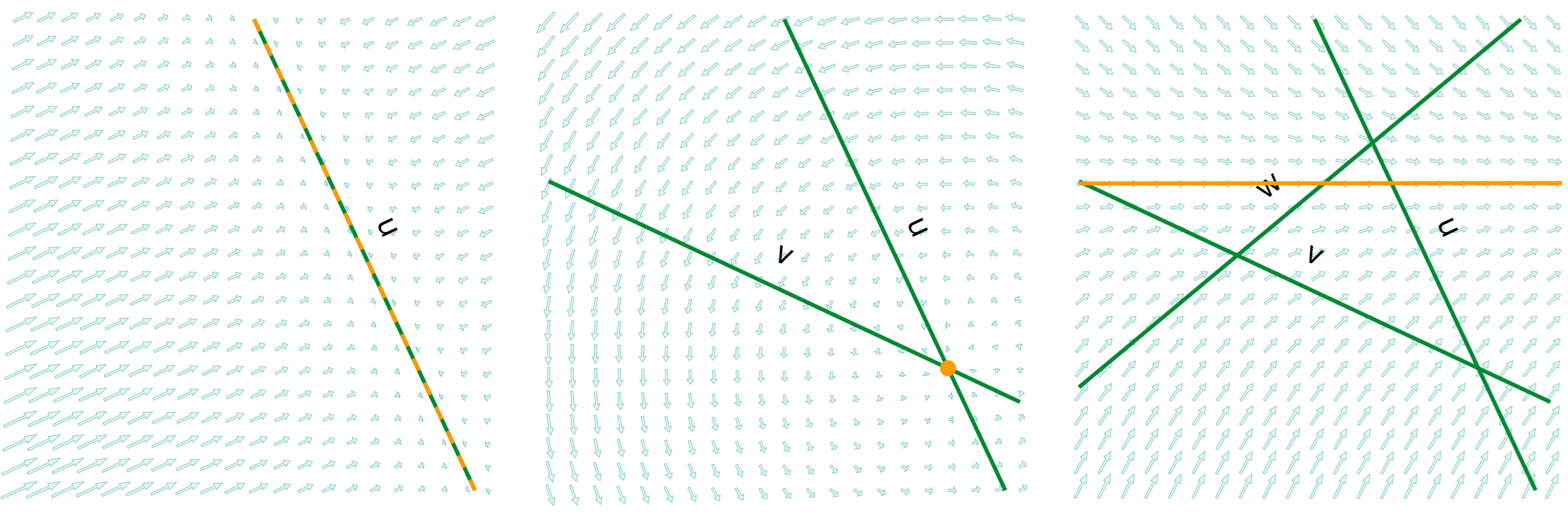}
    \caption{Reflections in 1, 2 and 3 lines (green), and their main invariants (orange). }
    \label{fig:reflections123}
\end{figure}

The factorisation of $k$-reflections into reflections is not unique when $k>1$. For example, \cref{fig:gauges} shows several possible factorisations of a bireflection. If the intersection point of -- and relative angle between -- the two reflection lines is kept constant, the same rotation will always result. That is, every two adjacent reflections in a $k$-reflection can be freely rotated, or \emph{gauged}, around their intersection point without changing the composition.   

\begin{figure}[h]
    \centering
    \includegraphics[width=\textwidth]{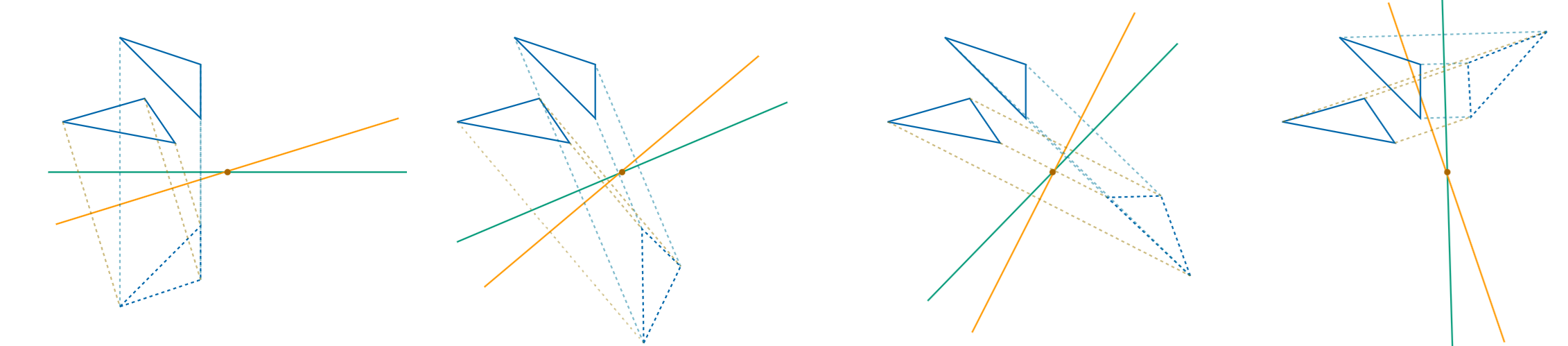}
    \caption{Four possible factorisations of a bireflection. Notice that the initial and final states are the same, and hence there is a gauge degree of freedom that can be chosen freely.}
    \label{fig:gauges}
\end{figure}

The geometric gauges can now be used to transform our initial factorisation from \cref{fig:reflections123}, into a (pseudo) canonical factorisation where we try to orthogonalize the reflections. 
Doing so reveals the \emph{invariant decomposition} of the trireflection $wvu$ in \cref{fig:transflections}, 
into an orthogonal (and commuting) reflection $w'$ and translation $v'' u'$ and hence explains the mysterious appearance of the orange invariant.

\begin{figure}[h]
    \centering
    \includegraphics[width=\textwidth]{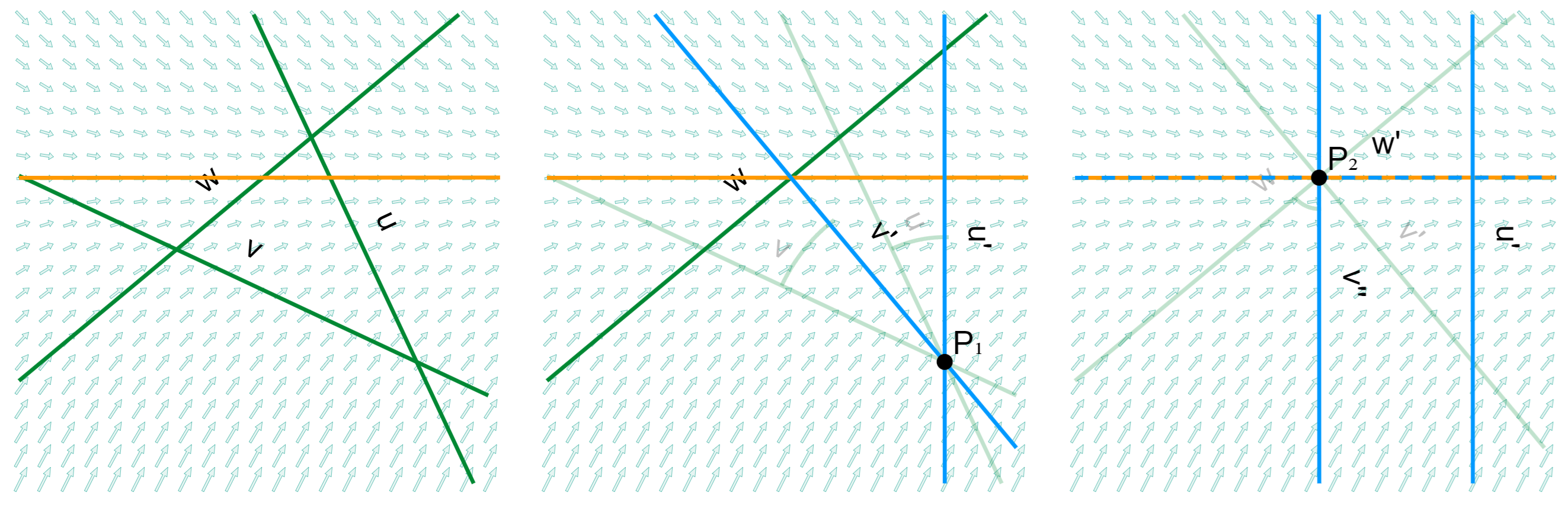}
    \caption{ From left to center: the gauge at $P_1$ is used to make $v' \perp w$. From center to right: the gauge at $P_2$ is used to make $w'\perp u'$. This reveals any composition of three reflections in the plane to always be a transflection, with invariant line $w'$. }
    \label{fig:transflections}
\end{figure}

\subsection{Double cover on the flip side}
\label{sec:double_cover}

\begin{figure}[]
    \centering
    \includegraphics[width=\textwidth]{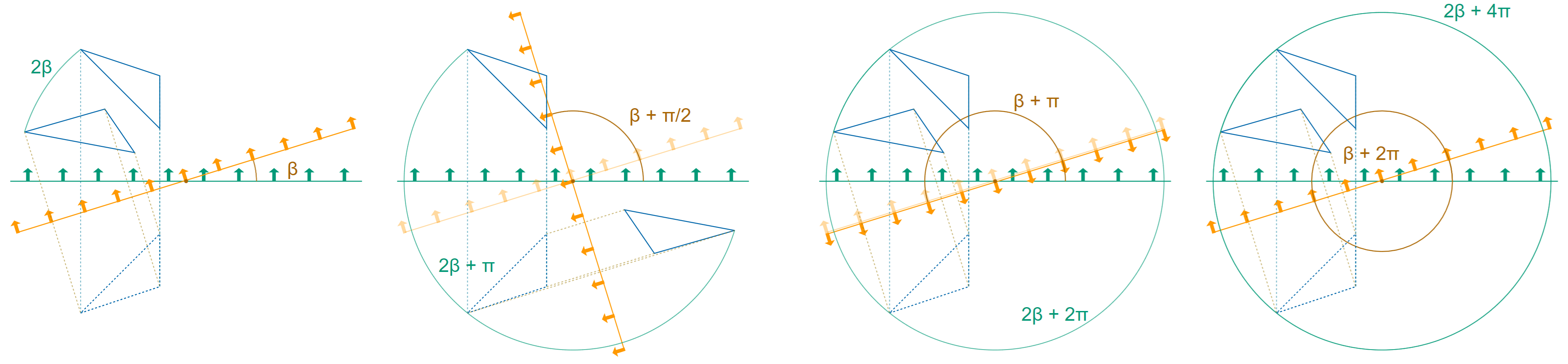}
    \caption{A triangle reflected in two intersecting mirrors separated by an angle $\beta$ rotates by $2\beta$. 
    Thus, when the angle between the mirrors is $\beta + \pi$ (third image), 
    the triangle has been rotated over $2 \beta + 2 \pi$, which is equivalent to a rotation over $2 \beta$. 
    But notice that the orange reflection is not back to where it started: instead, it faces the other way.
    The original scenario is restored only when the mirrors are at $\beta + 2 \pi$.
    }
    \label{fig:oriented_reflections}
\end{figure}

Despite explanations of rotations being ubiquitous in physics literature, 
their explanation is sadly always based on \SO{3} rather than \Spin{3}, 
which leads to an incomplete understanding of the true nature of rotations.
We start with the simplest case: a rotation around a point $B$ in a plane. 
Such a rotation is identical to two reflections in intersecting lines through $B$, and hence it is a \emph{bireflection}.
Denoting these reflections as $u$ and $v$, the bireflection $R = vu$ has the exponential form
    \begin{equation}\label{eq:evolution}
        R(\theta) = e^{\theta B} R_0.
    \end{equation}
\Cref{fig:oriented_reflections} shows what happens to a triangle $\blacktriangle$ that is being rotated around the point $B$ using conjugation:
    \[ \blacktriangle(\theta) = R \blacktriangle \widetilde{R}. \]
We start from the identity bireflection $R = 1 = u^2$, which corresponds to $\theta=0$ and $R_0 = 1$, and rotate $v$ until $u$ and $v$ are orthogonal, and hence $\theta=\pi / 2$.
Notice how the triangle has already rotated over $\pi$ by this point, which is obvious by looking at the constituent reflections.
If we now keep rotating $v$ until the angle between the reflections is $\theta = \pi$, then the triangle returns to its original position, but notice how the bireflection did not! 
Since we were careful to distinguish the front and the back of the reflections $u$ and $v$, it is clear that the bireflection did not yet return to its original position: the triangle has to be taken for another round before the bireflection is returned to its original position.
To any physicist, this will be reminiscent of how spinors are said to have this mysterious quality whereby they return to their original position only after 720 degrees of rotation. 
But there is nothing mysterious happening here: the bireflection itself behaves fully classical and is back to its original state once $v$ has been rotated over $2 \pi$.
By realizing that rotations are bireflections, and that reflections are oriented, the mystery disappears.

The confusion likely occurred because we are used to thinking about the object being rotated, not the object doing the rotation, which is in some sense its ``square root''. 
Perhaps what was lacking was a good mental model for directly visualizing the object performing the rotation itself, which is what bireflections offer.
And spinors are the things  doing the rotation, which is evidenced by them always acting on some operator, such as for example in expectation values of operators as $\langle \hat{O} \rangle = \expval{\hat{O}}{\psi}$, in the Dirac Lagrangian as $\mathcal{L} = \bar{\Psi} (i \partial - m) \Psi$, etc.
Indeed, if we replace $R$ by $\psi$ and $\theta$ by $t$  in the above equations, the equations start to look eerily familiar to anyone familiar with quantum mechanics.

It is often stated that \SO{n} is the rotation group, while \Spin{n} is its double cover, which is often seen as a purely algebraic concept with no geometrical interpretation. But as \cref{fig:oriented_reflections} shows, the opposite is true: \Spin{n} is the rotation group, while \SO{n} is its half-cover.
The mnemonic here is ``Double-sidedness leads to double-covers''.
Ignore the front and back of mirrors at your own peril.

When done correctly, classical rotations make for a much more illuminating model for the ``strange'' rotation behavior of spinors than the Dirac belt trick or the Filipino wine dance.
But the obvious question is now: is this bireflection a spinor? 
Yes, but it is only the $\ket{+}$ component. 
To find the $\ket{-}$ component, and to understand the generalizations to higher numbers of dimensions, 
we will need the invariant decomposition theorem.

\subsection{Invariant Decomposition}\label{sec:inv_decomposition}

We have seen that two reflections compose into either a rotation, translation, or boost (\cref{fig:bireflections}).
Moreover, we have seen that three reflections in the plane can be factored into a commuting reflection and bireflection (\cref{fig:transflections}).
This principle shows that there are no fundamentally new orthogonal transformations beyond bireflections.
The quadreflection analogue of this principle in 3D is known as the Mozzi-Chasles' theorem:
    \begin{displayquote}``Every three dimensional rigid body motion can be decomposed as a translation along a line followed or preceded by a rotation around the same line.''
    \end{displayquote}
Since every three dimensional rigid body motion is at most a quadreflection, and because both the aforementioned rotation and translation are bireflections, this statement can be put more generally as
    \begin{displayquote}``Every quadreflection can be decomposed into two commuting bireflections.''
    \end{displayquote}
Now suddenly the statement no longer seems dimension or signature dependent, and the generalization to arbitrary dimensions practically writes itself:
    \begin{theorem}[Invariant decomposition]\label{th:invariantdecomposition}
    A product $U=u_1 u_2 \cdots u_\ell$ of $\ell$ reflections $u_i$ can be decomposed into exactly $\ceil{\frac \ell 2}$ simple commuting factors. These are $\floor{\frac \ell 2}$ products of two reflections, and, for odd $\ell$, one extra reflection.
    \end{theorem}
    \begin{proof}
        \cite{GSG,outer}. 
    \end{proof}
The invariant decomposition as first introduced in \cite{GSG} and more rigorously described in \cite{outer}, implies that any $U \in \Spin{p, q, r}$, as the product of at most $2k$-reflections, can be decomposed into $k$ commuting bireflections. 
Each of these commuting bireflections is generated by a \emph{simple} bivector, meaning a 2-blade, which is an element of the Lie algebra \spin{p, q, r}.
Thus, the invariant decomposition has an important consequence:
    \begin{displayquote}
        ``Everything can be made to look simple.''
    \end{displayquote}
Let us denote the invariant factorization of $U \in \Spin{p, q, r}$ into mutually commuting bireflections $U_i$ as $U = U_1 U_2 \cdots U_k$, where $U_i = c_i + s_i \bl_i$ with $c_i, s_i \in \mathbb{R}$, and $\bl_i$ a simple bivector, which as a 2-blade squares to a scalar: $\lambda_i = \bl_i^2$. As the product of two orthogonal vectors, the simple bivectors $\bl_i$ are the invariant points that characterize each bireflection.

Eigenelements of $U$ under conjugation are those elements $W$ satisfying
    \begin{equation}
        U[W] \coloneqq U W U^{-1} = \gamma W.
    \end{equation}
In what follows we assume that $W = w_1 \cdots w_\ell$ is a product of (possibly null) vectors, given the clear geometric interpretation of such quantities.
From the fact that $U[W]$ is a norm preserving map, we find $W \widetilde{W} = \gamma^2 W \widetilde{W}$, 
leading to the conclusion
that eigenelements have either eigenvalue $\gamma = \pm 1$, or are null.
An eigenelement with $\gamma = \pm 1$ is not just an eigenelement of $U$, but it is an \emph{invariant} of $U$, meaning that it is either completely unaffected by the transformation, or has its orientation inverted at most.
Moreover, an eigenvalue of $\gamma = +1$ implies that $W$ commutes with $U$, while $\gamma = -1$ implies $W$ anti-commutes with $U$.

By the invariant decomposition, we know that $W$ is also an eigenelement of each bireflection $U_i = c_i + s_i \bl_i$ separately, which in turn implies that $W$ is also an eigenelement of the $\bl_i$ under commutation:
    \begin{equation}
            \bl_i \times W \coloneqq \tfrac{1}{2}\comm{\bl_i}{W} = \mu W
    \end{equation}
It is again straightforward to prove that all eigenelements of $\bl_i$ must either be null or have eigenvalue $\mu = 0$, by using the fact that $\bl_i \times (W \widetilde{W}) = 0$ \cite{outer}.

Since $W = w_1 \cdots w_\ell$, and by the property $U[W] \propto \prod_{j=1}^\ell U[w_j]$, it is sufficient to consider eigenvectors $w$ only.
There are three geometrically distinct cases\footnote{More exotic cases are possible when $\lambda \notin \mathbb{R}$, which can happen in signatures like \GR{2,2}, but for the current discussion we stick to Euclidean and Lorentzian signatures where $\lambda \in \mathbb{R}$. For more information on the general case see \cite{GSG, outer}.} to be considered, depending on the square of $\bl_i$:
    \begin{description}
        \item[Boost] When $\lambda_i = \bl_i^2 > 0$, the point $\bl_i$ is the generator of a boost (\cref{fig:bireflections}, left). Since $\bl = u v$ with $u^2 = - v^2 = \pm 1$, we find that the isotropic eigenvectors are given by
            \begin{equation}
                w_{\pm} = \tfrac{1}{2} (u \mp v).
            \end{equation}
        with eigenvalues $\mu_\pm = \pm 1$. From a geometric perspective, the lines $w_+$ and $w_-$ form the lightcone.
        \item[Translation] When $\lambda_i = \bl_i^2 = 0$, the point $\bl_i$ is the generator of a translation (\cref{fig:bireflections}, middle). Since $\bl = u v$ with $u^2 = 0$ and $v^2 \neq 0$, $u$ is an eigenvector of $\bl$ with eigenvalue $\mu = 0$. From a geometric perspective, the line $u$ is the horizon, which is indeed invariant under translation. In the context of spinors, translations will not be considered.
        \item[Rotation] When $\lambda_i = \bl_i^2 < 0$, the point $\bl_i$ is the generator of a rotation in a plane (\cref{fig:bireflections}, right). This implies $u^2 = v^2 = \pm 1$; we shall choose $u^2 = 1$. 
        In general rotations do not leave any line in the plane invariant except for the line at infinity, unless the angle of rotation happens to be a multiple of $\pi$. Thus, there are no solutions to $\bl \times w = \bl \cdot w = \mu w$ for $w \in \GR[(1)]{2}$. 
        If we insist however on having vectorial solutions to this equation, then there is no choice but to consider vectors $w \in \mathbb{C} \otimes \GR[(1)]{2} = \GC[(1)]{2}$. Now we can construct the isotropic eigenvectors
            \begin{equation}
                w_{\pm} = \tfrac{1}{2} (u \pm i v)
            \end{equation}
        with eigenvalues $\mu = \pm i$ assuming $\bl$ is normalized such that $\bl^2 = -1$.
        So in order to find eigenvectors of $\bl$ under commutation, we have no choice but to introduce complex numbers.
        But we should not forget that we are using the commutator of $\bl$ with vectors as a proxy for finding eigenelements $W$ under conjugation by the bireflection $U = \cos(\theta / 2) + \bl \sin(\theta / 2) $.
        The null solutions $w_\pm$ are still valid eigenvectors under conjugation, satisfying $U[w_\pm] = e^{i \theta} w_\pm$, but there are now also real solutions! 
        Firstly, any element $W = \alpha + \beta \bl$ commutes with $U$ for all $\theta$ and is therefore an eigenelement with $\gamma=1$. For general $\theta$ there is no solution corresponding to $\gamma = -1$, but in the special case that $U = \bl$, any line $w = \alpha u + \beta v$ that intersects with the point $\bl$ anti-commutes with $\bl$, and thus satisfies $\bl[w] = - w$.
        So only in the special case of $U = \bl$, we can use conjugation to divide the subspace defined by the point $\bl$ into a commuting and an anti-commuting subspace.
    \end{description}

As the above discussion on rotations highlights, 
there are two ways to define eigenelements under $U \in \Spin{p,q,r}$:
\begin{itemize}
    \item Construct eigenelements $W$ from the eigenvectors of the generators $\bl_j$ under the commutator product: $\bl_j \times w_j = \mu_j w_j$.
    \item Consider eigenelements under point-reflection in $\bl_j$: $\bl_j[W] = \gamma_j W$. 
\end{itemize}
The traditional approach to spinors follows the path of null eigenvectors under commutation with the $\bl_j$ to construct the spinor basis, as we will review in \cref{sec:traditional_spinors}. 
Contrarily, we propose to build a basis out of the geometric invariants under point reflections in $\bl_j$ and moreover, claim that this is the information the traditional spinor basis is trying to capture in the first place. This will be the subject of \cref{sec:geometric_spinors}.

At this point the qualification of basis spinors using the double-sided conjugation by $\bl_j$ might still cause some discomfort, because spinors transform one-sidedly, and thus one might argue that they should also be labelled under a one-sided operation.
But as we have shown in \cref{sec:double_cover}, \cref{eq:evolution}, ordinary rotations also evolve one-sidedly, yet the bireflection $R(t) = \exp(t \bl)$ can be labelled a $\gamma = +1$ invariant for $\bl$ under conjugation.
So labelling can certainly be performed with the double-sided conjugation, despite spinors evolving one-sidedly.

Since the $\gamma = -1$ states do not exist for arbitrary $U_j$, and instead only exist for the special case of point-reflection in the generators $U_j = \bl_j$,
it might seem somewhat arbitrary to single out point reflection in $\bl_j$ in order to distinguish up and down eigenstates states.
But in fact there is very good argument for doing so, and it ties in nicely with the local geometric gauge degrees of freedom that every point has.

\subsection{What is the point?} \label{sec:point}

In order to understand what is special about point reflections, we first have to understand the point.
In a pseudo-Euclidean space of $d = p + q$ dimensions, vectors represent $d-1$ dimensional subspaces called hyperplanes, 2-blades represent $d-2$ dimensional subspaces called hyperlines, ..., $(d-1)$-blades represent 1-dimensional lines, and $d$-blades represent 0 dimensional points.
\Cref{fig:elements_of_geometry} shows this for $d=3$.
The justification for the identification of a point with a $d$-blade is a purely geometric one: a $d$-blade is equivalent to the product of $d$ orthogonal  hyperplanes, and reflection in $d$ orthogonal hyperplanes is identical to a point reflection in their unique intersection point. 
A point therefore \emph{is} the product of $d$ orthogonal hyperplanes.
After all, if it walks like duck, and it quacks like a duck, it is a duck.
Moreover, as we have seen in \cref{sec:double_cover}, hyperplanes have orientation, and therefore so do points!
And lastly and perhaps most importantly, the particular factorization into orthogonal hyperplanes is not unique: there is an entire $\Spin{p, q}$ \emph{local gauge group} at each point!
Points are therefore not some boring zero dimensional quantity that is hardly worth talking about, 
but instead 
\begin{displayquote}
    ``Points are the peak of a mountain, and have internal structure, local gauge degrees of freedom, and orientation.''
\end{displayquote}
Any point $O_d$ in a $d$-dimensional pseudo-Euclidean space can be factorized into $d$ orthogonal hyperplanes $v_i$, or into $k = \floor{d / 2}$ orthogonal hyperlines $\bl_j$:
\begin{equation}\label{eq:point}
    O_d = \prod_{i=1}^d v_i = \begin{cases}
        \qquad\;\; \prod_{j=1}^k \bl_j & d \text{ even} \\
        v_{2k+1} \prod_{j=1}^k \bl_j & d \text{ odd}
    \end{cases} \; .
\end{equation}
The point $O_d$ therefore behaves like the pseudoscalar in the space defined by it, and each of the $\bl_j$ in turn behaves like a pseudoscalar in the 2D subspace it defines.
Using the mutually orthogonal $\bl_j$, 
the entire space defined by $O_d$ can be divided into segments which either commute or anti-commute with each of the $\bl_j$.
In doing so, any element $\zeta \in \GR{p,q}$ of the tangent space at $O_d$ can be decomposed into 
\[ \zeta = \sum_{\vec{s}} \zeta_{\vec{s}}\, , \]
where $\vec{s} = (s_1, s_2, \ldots, s_k)$ with $s_j = \pm 1$ denotes the eigenvalues of $\zeta_{\vec{s}}$ under conjugation with each $\bl_j$.
In other words,
    \[ \bl_j[\zeta_{\vec{s}}] = \bl_j \, \zeta_{\vec{s}} \ \bl_j^{-1} = s_j \zeta_{\vec{s}} \,. \]
This binary labeling is one of the key features of spinors, and the fact that it follows naturally from the factorisation of the point $O_d$ is extremely interesting.
Moreover, when $d$ is even the point $O_d$ commutes with all elements of the even subspace while it anti-commutes with every element of the odd subspace, and therefore it can be used to split any element $\zeta$ into a commuting and anti-commuting part:
    \begin{equation}
        \zeta_{L/R} = \tfrac{1}{2} (\zeta \pm O_d \zeta O_d^{-1}).
    \end{equation}
When $d$ is odd however, $O_d$ commutes with every element in the subspace and can therefore no longer be used to divide the subspace into two equal parts. 
However, \cref{eq:point} makes it evident that $O_d = v_{2k+1} O_{2k}$, and thus $O_{2k} = \prod_{j=1}^k \bl_j$ can still be used to divide the subspace into two equal parts.

In conclusion, rather than being zero dimensional quantities that hardly need any explanation, points have a lot of features that we traditionally use when constructing spinor representations.
To understand this relationship better,
we should have a look at the traditional approach to spinors, before we can continue with a geometrical approach to spinors.

\section{Spinors}\label{sec:spinors}
In this section we will first discuss traditional spinors, and will then proceed to introduce \emph{pointors}.
By traditional spinors, we are referring to the definition of spinors as elements of a minimal left ideal, upon which the spin group acts from the left, also referred to as algebraic spinors.
This definition goes back to Cartan, Weyl, etc., and is widely adopted in the mathematical physics community, see e.g.\cite{Polchinski:1998rr,Vaz:2016qyw,Lou,jost2008riemannian}.
On the face of it this definition is something quite different from the traditional view of spinors within the GA community, as members of the even subalgebra which act grade preservingly on vectors; a definition due to Hestenes \cite{GeometricCalculus}.

In an attempt to bridge this gap we take a different approach from other authors such as e.g. \cite{FRANCIS2005383}:
we show from first principles how some characteristic properties of spinors of arbitrary dimensions can be captured by real geometrical structures we call \emph{pointors}, which are the sum of an even and an odd versor.
This reveals that Hestenes' definition of spinors tells only half the story: the even half.

Although pointors have a lot in common with traditional spinors, and share many of their properties, they are not the same thing for $d > 3$.
It is for this reason that we chose to introduce them under a new name: so as not to add to the confusion around the proper definition of a spinor.
In future work we will show how to constrain a pointor to a spinor.

We start this section with a trigger warning:
in this section, we will sin against our philosophy of dimension and signature agnosticism by only considering all positive or all negative signatures (\GR{d} or \GR{0,d}). We will atone for this sin in future work on this subject, but since the goal of the current work is to give an introduction, this sin should be forgivable.

\subsection{Algebraic Spinors}\label{sec:traditional_spinors}

In order to build spinor representations, one
ordinarily starts by choosing the dimension of the space $d$, which is either an even number $d = 2k$ or an odd number $d = 2k + 1$. One then selects a set of $k$ commuting elements $\bl_j$ of the Lie algebra; the Cartan subalgebra, and one takes the corresponding set of eigenvectors $w_{\pm j} = \tfrac{1}{2}(u_j \pm i v_j)$. Recall that these eigenvectors are necessarily isotropic (null), as discussed in \cref{sec:inv_decomposition}. In odd dimensions there is one extra (possibly non-null) basisvector $v_{2k + 1}$ which commutes with all $\bl_j$.
Using all the isotropic eigenvectors, one constructs the master idempotent
    \begin{equation}
        \boxplus \coloneqq \prod_{j=1}^k w_{+j} w_{-j}.
    \end{equation}
One then defines the set of normalized bivectors $\beta_j = \bl_j / \mu_j$, where $\mu_j$ is the eigenvalue of $w_{+j}$. This ensures that the $\beta_j$ are all normalized to satisfy $\beta_j^2 = 1$.
The idempotent $\boxplus$ has eigenvalue $+1$ for every $\beta_j$ since $\beta_j \boxplus = + \boxplus$, and is annihilated by multiplication with any $w_{+j}$ from the left: $w_{+j} \boxplus = 0$.
We can now form spinor basis states $\eta_{\vec{s}}$ with eigenvalues $\vec{s} = ( s_1, \ldots, s_k )$, where $s_j = \pm 1$, by acting on $\boxplus$ with the lowering operators $w_{-j}$:
    \begin{equation}
        \eta_{\vec{s}} = \bqty{\prod_{j=1}^k w_{-j}^{(1 - s_j) / 2}} \boxplus.
    \end{equation}
Due to the isotropy of the basisvectors $w_{\pm j}$, all $\eta_{\vec{s}}$ satisfy $\widetilde{\eta}_{\vec{s}} \ \eta_{\vec{s}} = 0$.
Moreover, the $\eta_{\vec{s}}$ form a basis for the spinor space $S$, whose elements $\eta \in S$ are a linear combination of the form 
\[ \eta = \sum_{\vec{s}} c_{\vec{s}} \ \eta_{\vec{s}},\] 
where $c_{\vec{s}} \in \mathbb{C}$ and where the sum is understood to range over all $2^k$ possible $\vec{s}$. 
The product of all $\beta_j$ defines the chiral operator $\Gamma$:
    \begin{equation}
        \Gamma \coloneqq \prod_{j=1}^k \beta_j \; ,
    \end{equation}
which satisfies $\Gamma^2 = 1$, commutes with all $\beta_j$, and anti-commutes with all vectors $w_{\pm j}$ but (in odd dimensions) commutes with $v_{2k + 1}$.
Therefore the chiral operator $\Gamma$ is simply a (complex) scalar multiple of the point $O_{2k}$.
The chiral operator $\Gamma$ distinguishes between those spinor states $\eta_{\vec{s}}$ with an even or odd number of $s_j = -1$, which can be used to form the chiral projectors
    \[ P_{L/R} = \tfrac{1}{2} (1 \pm \Gamma). \]
Using these projectors, $\eta$ can be split into $\eta_{L/R} = P_{L/R} \eta$, known as the left and right Weyl spinors.
There are several important observations to make:
\begin{enumerate}
    \item Basis spinors $\eta_{\vec{s}}$ are essentially binary labels, characterized by a series of $\pm 1$ eigenvalues under the multiplicative action of the bivectors $\beta_j$.
    \item The spinor space $S$ is a $2^{k}$-dimensional complex vector space ($2^{k+1}$ over the reals).
    \item The spinor space $S$ includes an equal number of even and odd elements, and can be split into the equally sized spaces $S^{\pm}$ using the projectors $P_{L/R}$, each of which is a $2^{k-1}$-dimensional complex vector space ($2^{k}$ over the reals).
\end{enumerate}
The question now is: can we preserve these properties by walking the path of real geometrical invariants, instead of the imaginary null path?

\subsection{Back to Reality: Pointors} \label{sec:geometric_spinors}
With the traditional view of spinors in the back of our mind, it is time to hit the road to reality.
Let us pause on the eigenvectors $w_{\pm}$ of a simple bivector $\bl = u v$ one more time, where $u^2 = v^2 = 1$ and $u \cdot v = 0$. 
Any vector $r = \cos(\theta_-) u + \sin(\theta_-) v$ with $\theta_- \in \mathbb{R}$ satisfies $\bl \wedge r = 0$, and can be decomposed in the eigenbasis as 
\[ r = e^{-i \theta_-} w_+ + e^{i \theta_-} w_- \ .\] 
Similarly, any bireflection $R = e^{\theta_+ \bl}$ 
can be written in the eigenbasis as 
    \[ R = e^{i \theta_+} w_+ w_- + e^{-i \theta_+} w_- w_+ \ . \]
To describe spinors in 2D, we would take $\psi = \rho_- r + \rho_+ R$
and slam it into the master idempotent $w_+ w_-$, resulting in 
    \begin{equation}
        \Psi = (\rho_- r + \rho_+ R) w_+ w_- = c_- w_- + c_+ w_+ w_- \,
    \end{equation}
where $c_\pm = \rho_\pm e^{i \theta_\pm}$.
From this it is clear that the complex null spinor $\Psi$ represents the exact same degrees of freedom as the real object $\psi$. 
Moreover, $R$ and $r$ are two eigenstates of $\bl$ under conjugation, since $\bl[R] = + R$ and $\bl[r] = -r$, corresponding to $\bl \times w_+ w_- = + i w_+ w_-$ and $\bl \times w_- = - i w_-$ respectively in the traditional picture.
The positive and negative eigenstates under conjugation in $\bl$ can be extracted as
    \[ \psi_\pm = \tfrac{1}{2} (\psi \pm \bl \psi {\bl}^{-1}).\]
So in the 2D case we now have all the properties we want: two eigenstates under $\bl$ that can serve as binary labels, and two lots of two real degrees of freedom, making a total of four real (two complex) degrees of freedom.
The object $\psi = \rho_- r + \rho_+ R$ 
can affect the magnitude of the point $\bl$ because it satisfies $\psi \bl \widetilde{\psi} \propto \bl$. As such we call it a \emph{pointor}.

In order to make pointors more closely resemble the traditional spinors with complex coefficients,
we observe that
the bivector $\bl$ behaves just like an imaginary unit,
so the complex numbers $c_\pm = \rho_\pm e^{i \theta_\pm}$ can be replaced by the real bireflections $\rho_\pm e^{\bl \theta_\pm}$.
Additionally, in order to measure anything, we need to bring a measuring stick. In other words, we need a reference in both the even and odd subalgebra. We will chose $\psi_-^{\text{ref}} = u$ and $\psi_+^{\text{ref}} = u^2 = 1$, where $u$ is any normalized vector satisfying $\bl \wedge u = 0$.
Now any pointor $\psi = \rho_- r + \rho_+ R$ is of the form
    \begin{equation}
        \psi = \rho_- e^{\bl \theta_-} \psi_-^{\text{ref}} + \rho_+ e^{\bl \theta_+} \psi_+^{\text{ref}}.
    \end{equation}
This form looks a lot more similar to the $\ket{\psi} = c_- \ket{-} + c_+ \ket{+}$ from quantum mechanics textbooks, although the pointor $\psi$ is still a flatlander whereas $\ket{\psi}$ is typically defined in 3D space.

Using operations in the 2D plane in which this flatland pointor lives, it is impossible to mix the positive and negative components.
However, in the presence of a third dimension, the positive and negative components can be mixed, for example by a rotation around $v_3 u$ where $v_3$ is the new plane orthogonal to flatland. 
The presence of a third dimension does not change anything about the flatland spinor qualitatively however; just like a single bireflection in 2D or in 3D is exactly the same thing, but it does mean the degrees of freedom become harder to isolate.
In order to keep the degrees of freedom easy to manipulate, 
we can use the vector $v_3$ orthogonal to $\bl = O_2$ to lift the entire pointor into the even subalgebra of 3D space by forming
    \begin{equation}
        \psi = \rho_- e^{\bl \theta_-} \psi_-^{\text{ref}} v_3 + \rho_+ e^{\bl \theta_+} \psi_+^{\text{ref}}.
    \end{equation}
It is straightforward to show that such a pointor is a quaternion, and is identical to Pauli spinors \cite{GA4Ph,FRANCIS2005383}.

From 4D and up we have something new, because there are at least two commuting $\bl_j$ that form the point $O_4 = \bl_1 \bl_2$. 
We therefore have reference reflections $\psi_{-+}^{\text{ref}} = u_1$ and $\psi_{+-}^{\text{ref}} = u_2$ such that we get the reference states
    \begin{equation}
        \psi_{++}^{\text{ref}} = 1, \quad \psi_{-+}^{\text{ref}} = u_1, \quad \psi_{+-}^{\text{ref}} = u_2, \quad \psi_{--}^{\text{ref}} = u_1 u_2,
    \end{equation}
which are rotated into place by ``hyper-complex'' numbers $\rho \exp(\theta_1 \bl_1 + \theta_2 \bl_2)$.
For convenience we will label these reference states $\psi_{\vec{s}}$, where $\vec{s} = ( s_1, s_2 )$ is again a tuple of eigenvalues that can be used to label the states.
A general 4D pointor is of the form
    \begin{equation}\label{eq:4D_pointor}
        \psi = \sum_{\vec{s}} \psi_{\vec{s}} = \sum_{\vec{s}} \rho_{\vec{s}} \, e^{\sum_j \theta_{\vec{s} j} \bl_j} \psi_{\vec{s}}^{\text{ref}} \, . 
    \end{equation}
It is important to note that flatland spinors are included in this definition, and are therefore still valid pointors in 4D, something demanded by the first postulate, but which is not true in the traditional view of spinors since e.g. $w_{+1} w_{-1}$ is not an eigenstate of $\bl_2$ under left multiplication.

The question naturally arises if there is an easier way to think about \cref{eq:4D_pointor}, which leads to the following theorem:
\begin{theorem}
        For $k < 3$,
            \[ \psi = \sum_{\vec{s}} \psi_{\vec{s}} = \rho_+ R + \rho_- P, \]
        where $\rho_\pm \in \mathbb{R}$, $R \in \Spin{p,q}$ and $P \in \Pin[-]{p,q}$. 
    \end{theorem}
    \begin{proof}
        The even and odd parts can be treated separately, and are defined as 
        \[ \psi_{L/R} = \tfrac{1}{2} (\psi \pm O \psi O^{-1}).\]
        We have to show that any even/odd versor can be written as a linear combination of invariants $\psi_{\vec{s}}$, and conversely that any $\psi_{L/R}$ is an even/odd versor.
        We only present the proof for $\psi_L$, since the proof for $\psi_R$ follows \emph{mutatis mutandis}.

        For the first part of the proof, assume that $\psi_L = \rho R$ is a versor, with $\rho \in \mathbb{R}$ and $R \in \Spin{p,q}$ a $2k$-reflection. We already know that $R$ permits an invariant decomposition as described in \cref{sec:inv_decomposition}, which yields the natural set of $\bl_j$ to decompose the point $O$.
        We now immediately know that $\psi = \psi_{++}$ with all positive eigenvalues, and therefore it is possible to express any versor as $\psi = \sum_{\vec{s}} \psi_{\vec{s}}$.

        For the second part of the proof, we use that $\psi_L$ permits a polar decomposition into $\psi_L = S R$, where $R \in \Spin{p,q}$ is a quadreflection, and $S$ is an even self-reverse element. 
        To prove that $\psi_L$ is a versor, it therefore suffices to show that $\psi_L \widetilde{\psi}_L = S^2 \in \mathbb{R}$. 
        Since $\psi_L = \psi_{\scriptscriptstyle ++} + \psi_{\scriptscriptstyle --}$, and each $\psi_{\vec{s}}$ has a scalar norm, we need to show that the cross terms are zero, i.e.
        \[ (\psi_{\scriptscriptstyle ++} + \psi_{\scriptscriptstyle --}) \widetilde{(\psi_{\scriptscriptstyle ++} + \psi_{\scriptscriptstyle --})} 
        \stackrel{?}{=} \psi_{\scriptscriptstyle ++} \widetilde{\psi}_{\scriptscriptstyle ++} + \psi_{\scriptscriptstyle --} \widetilde{\psi}_{\scriptscriptstyle --}
        \implies 
        \psi_{\scriptscriptstyle ++} \widetilde{\psi}_{\scriptscriptstyle --} + \psi_{\scriptscriptstyle --} \widetilde{\psi}_{\scriptscriptstyle ++} = 0.\]
        But there exists parameters $\rho, \delta_1, \delta_2 \in \mathbb{R}$ such that
            \begin{align*}
                \psi_{\scriptscriptstyle ++} = \rho e^{\delta_1 \bl_1 + \delta_2 \bl_2} u_1 u_2 \psi_{\scriptscriptstyle --} \, .
            \end{align*}
        We then find by direct computation
            \begin{align*}
                \psi_{\scriptscriptstyle ++} \widetilde{\psi}_{\scriptscriptstyle --} + \psi_{\scriptscriptstyle --} \widetilde{\psi}_{\scriptscriptstyle ++} 
                &= \pqty{\rho e^{\delta_1 \bl_1 + \delta_2 \bl_2} u_1 u_2 + \rho u_2 u_1 e^{-\delta_1 \bl_1 - \delta_2 \bl_2}} \psi_{\scriptscriptstyle --} \widetilde{\psi_{\scriptscriptstyle --}} \\  
                &= \rho \, e^{\delta_1 \bl_1 + \delta_2 \bl_2} \pqty{ u_1 u_2 + u_2 u_1} \psi_{\scriptscriptstyle --}  \widetilde{\psi_{\scriptscriptstyle --}} = 0,
            \end{align*}
        where the last equality follows because $u_1 \cdot u_2 = 0$ by definition. Consequently $\psi_L \widetilde{\psi}_L \in \mathbb{R}$ when $k < 3$ and thus $d < 6$.
    \end{proof}
In 5D, the extra vector $v_5$ orthogonal to $O_4$ can again be used to lift everything up to the even subalgebra.

When we generalize to $k \geq 3$ however, we can no longer claim that $\psi = \sum_{\vec{s}} \psi_{\vec{s}}$ is a suitable generalization, due to the need for extra constraints to ensure cancellation of cross terms. 
This brings us to the formal definition of a \emph{pointor}.
\begin{definition}[Pointor]
    A pointor is defined to be an element $\psi$ in the subspace defined by the point $O$ which satisfies
        \[ \psi O \widetilde{\psi} = \rho O, \] 
    where $\rho \in \mathbb{R}$.
    Moreover, $\psi$ is a unit-pointor if $\rho = 1$.
\end{definition}
In other words, pointors leave points in place, but can potentially affect their magnitude. 
The following theorem gives the sought-after generalization of pointors to arbitrary dimensions.
\begin{theorem}
    Pointors are elements of the form
        \[ \psi = \rho_+ R + \rho_- P, \]
    where $\rho_\pm \in \mathbb{R}$, $R \in \Spin{p,q}$ and $P \in \Pin[-]{p,q}$. 
\end{theorem}
\begin{proof}
We start again from the polar decomposition $\psi = SR$, where $S$ is a self-reverse element and $R \in \Spin{p,q}$. We then find by direct computation
\begin{align}
    \psi O \widetilde{\psi} = S R O \widetilde{R} S = SOS \, .
\end{align}
In order for $SOS$ to be a scalar multiple of $O$, we must have $O S = \bar{S} O$, where $\bar{S}$ is the conjugate of $S$, such that $S \bar{S} \in \mathbb{R}$. 
If $S = \expval{S} + \expval{S}_1$, then  $O \expval{S}_1 = -\expval{S}_1 O$ and thus such an $S$ will satisfy $SOS = \small( S \bar{S} \small) O \propto O$.
However, a general $S = \expval{S} + \expval{S}_1 + \expval{S}_4 + \expval{S}_5 + \ldots$ will not satisfy $O S = \bar{S} O$, and thus we find that $\psi = S R$ is the sum of a scalar and vector multiple of $R$.
\end{proof}

The definition of a pointor is motivated by our observation that the point is the effective pseudoscalar of a tangent space defined by it, and it is therefore natural to ask which elements leave the point invariant.

A pointor can be decomposed into left and right Weyl pointors as
    \[ \psi_{L/R} = \tfrac{1}{2} (\psi \pm O \psi O^{-1}), \]
and even further decomposed into eigenstates $\psi_{\vec{s}}$ under conjugation with all the $\bl_j$, as shown in \cref{sec:point}. 
These eigenstates are of the form 
    \[ \psi_{\vec{s}} = \rho_{\vec{s}} \, e^{\sum_j \theta_{\vec{s} j} \bl_j} \psi_{\vec{s}}^{\text{ref}} \ . \]

\subsection{A Problem in Accounting?}\label{sec:accounting}
The astute reader might have noticed a significant problem in accounting. While traditional basis spinors $\eta_{\vec{s}}$ come with complex coefficients, basis point\-ors $\psi_{\vec{s}}$ come with ``hyper-complex'' coefficients $\rho_{\vec{s}} \exp\small(\sum_j \theta_j  \bl_j\small)$. So for $k \geq 2$ the pointor $\psi$ has more degrees of freedom than the traditional spinor $\Psi$.

There are several ways out of this conundrum. 
Firstly, we have to entertain the possibility that historically we mistook all the different ``imaginary units'' $\bl_j$ for the same imaginary unit $i$. 
After all, they all square to $-1$ and are mutually commuting, and if it walks like a duck, and it quacks like a duck... 
We observed the same thing in \cite{outer}, where we published a ``square root'' of the Cayley-Hamilton theorem, which was to the best of our knowledge previously unknown, and which leverages the same observation.
The success of proponents of the Hestenes definition of spinors in explaining a wide range of physics phenomena certainly suggests that life without complex minimal left ideals is possible.

Secondly, if complex minimal left ideals are reintroduced it is possible to force a loss of the distinction between the $\bl_j$, and to consider them all as the imaginary unit $i$. This can be achieved by throwing a pointor $\psi$ into the maw of the traditional master idempotent $\boxplus$: $\Psi = \psi \ \boxplus$, which works because 
\[ (\sum_j \theta_j  \bl_j) \boxplus = i (\sum_j \theta_j) \boxplus. \]
In this scenario the pointor $\psi$ can be thought of as the closest geometrically sensible object that can reproduce any traditional spinor $\Psi$, and therefore one can still think about the geometry of any spinor $\Psi$ as being that of a pointor $\psi$, modulo the idempotent $\boxplus$. 

Lastly, 
there is a way to achieve the same result with a specific \emph{real} idempotent.
This idempotent and its properties will be the subject of an upcoming paper, because it has far reaching consequences well beyond the scope of the current work.

\subsection{Hestenes Spinors}\label{sec:hestenes_spinors}
The key distinction between Hestenes' spinors and pointors is the fact that Hestenes prioritizes grade preservation on vectors, whereas pointors are defined to preserve points.
To be precise, in \cite[Section 3.8]{GeometricCalculus}, a spinor is defined as an element $\phi$ satisfying 
    \[ \phi x \widetilde{\phi} = \rho y, \]
where $\rho \in \mathbb{R}$  and $x, y \in \GR[(1)]{p,q}$ are vectors satisfying $x^2 = y^2$. 
By contrast, we defined a pointor $\psi$ as an element in the subspace defined by a point $O$ which satisfies
    \[ \psi O \widetilde{\psi} = \rho O \, .\]
At first glance the geometric narratives are strikingly similar, because in the point based worldview that Hestenes' spinors are defined in, vectors are used to point at points, and hence both of these equations attempt to describe how spinors affect points.
But because careful geometric analysis has revealed points in $d$-dimensional pseudo-Euclidean spaces to be identified with $d$-blades, the consequences of the two viewpoints are wildly different.
Whereas Hestenes' spinors rotate the point $x$ to the point $y$ and scale its magnitude by $\rho$, pointors instead leave the point where it is and only scale its magnitude.
Future research will explore the consequences of this significant difference.

\section{Conclusion \& Future Research}

Careful geometric analysis has revealed points in $d$-dimensional pseudo-Eucli\-dean spaces \GR{p, q, 1} to be $d$-blades.
This has far reaching consequences: it means that points can be factored into $d$ orthogonal hyperplanes and hence points are the pseudoscalars of a local \GR{p,q} algebra.
Moreover, this factorization is not unique, and hence points have a local $\Spin{p,q}$ geometric gauge group.
Furthermore, the factorization of a point defines a commuting set of $k = \floor{d / 2}$ bivectors, which 
is exactly what traditional descriptions of spinors use to label basis spinors.

This raised the suspicion that spinors might just be those elements which leave points invariant, which led us to the definition of \emph{pointors}: the sum of an even and odd versor which can only affect the magnitude of points but leave their location unchanged.
But space is full of points, and at each point there can be a non-trivial pointor affecting its magnitude.
This sounds a lot like a field theory already, and we got there merely by thinking about points.

Moreover, ordinarily the gauge groups of the standard model are seen as abstract internal degrees of freedom, without any geometrical motivation.
But this work reveals that points do in fact have \emph{geometrical} internal gauge degrees of freedom!
Future work will have to show if (some of) the internal gauge degrees of the standard model are of the geometric kind.

We were careful to distinguish pointors from spinors, because as discussed in \cref{sec:accounting}, pointors have more degrees of freedom than traditional spinors.
Although we expect pointors to be useful in their own right, we reserve the term spinor for pointors that have the traditional number of degrees of freedom.
In upcoming work we will show a novel way to construct spinors from pointors by throwing away the distinction between all $\bl_j$ in an unexpected way.

\section*{Acknowledgement}
The authors would like to thank the organizers of the ``Clifford Algebra and the Fundamental Forces of Nature'' session at the ICCA13 conference for the invitation to present aspects of this work as well as for organizing a stimulating conference.

\bibliographystyle{spmpsci}
\bibliography{biblio.bib}

\end{document}

%% file: definitions.tex
\usepackage{physics}
\usepackage{amsthm}
\usepackage{amssymb}
\usepackage{mathtools}  
\usepackage{csquotes}

\newcommand{\bl}{b}

\newcommand{\GR}[2][{}]{\ensuremath{\mathbb{R}^{{#1}}_{#2}}}
\newcommand{\GC}[2][{}]{\ensuremath{\mathbb{C}^{{#1}}_{#2}}}

\newcommand{\Spin}[2][{}]{\ensuremath{\textup{Spin}^{#1}({#2})}}

\newcommand{\Pin}[2][{}]{\ensuremath{\textup{Pin}^{#1}({#2})}}

\newcommand{\SO}[1]{\ensuremath{\text{SO}({#1})}}
\newcommand{\spin}[1]{\ensuremath{\mathfrak{spin}({#1})}}

\DeclarePairedDelimiter\floor{\lfloor}{\rfloor}
\DeclarePairedDelimiter\ceil{\lceil}{\rceil}